\documentclass[11pt]{article}
\usepackage[margin=1in]{geometry}
\usepackage{subfigure}
\usepackage{lineno}

\usepackage{amssymb}
\usepackage{amsmath}
\usepackage{color}
\usepackage[textsize=footnotesize]{todonotes}
\graphicspath{{figures/}} 

\newtheorem{theorem}{Theorem}
\newtheorem{lemma}{Lemma}

\newtheorem{corollary}{Corollary}

\newcommand{\remove}[1]{}

\newenvironment{proof}{{\em Proof.}}{\hspace*{\fill}$\Box$\vspace{2mm}}

\newcommand{\seg}[1]{\overline{#1}}

\def\comment#1{}

\def\ITEMMACRO #1 ??? #2 ???{\par\vskip2pt\noindent%
\hangindent=#2em\setbox0\hbox{#1\kern4pt}%
\ifdim\wd0<\hangindent\setbox0\hbox to\hangindent{\hss#1\kern.66em}\fi%
\box0\ignorespaces}

\newcommand{\T}{\mathcal{T}}

\begin{document}

\title{How many vertex locations can be\\ arbitrarily chosen when drawing planar graphs?}

\author{
    Emilio Di Giacomo\footnote{University of Perugia, Italy}
    \and
    Giuseppe Liotta\footnotemark[1]
    \and
    Tamara Mchedlidze\footnote{Karlsruhe Institute of Technology (KIT), Germany}
}

\date{}

\maketitle

\begin{abstract}
It is proven that every set $S$ of distinct points in the plane with
cardinality $\left \lceil \frac{\sqrt{\log_2 n}-1}{4}
\right \rceil$ can be a subset of the vertices of a crossing-free
straight-line drawing of any planar graph with $n$ vertices. It is
also proven that if $S$ is restricted to be a one-sided convex point
set, its cardinality increases to $\left \lceil \sqrt[3]{n} \right \rceil$. The
proofs are constructive and give rise to
$O(n)$-time drawing
algorithms. As a part of our proofs, we show that every maximal
planar graph contains a large induced biconnected outerplanar graphs
and a large induced outerpath (an outerplanar graph whose weak dual
is a path).

%

\end{abstract}

\thispagestyle{empty}
\newpage
\pagenumbering{arabic}

\section{Introduction and Overview}\label{se:introduction}

Computing drawings of planar graphs is an extensively studied
subject. In its most classical setting, the problem is to design an
algorithm that receives as input  a planar graph $G$ and it produces
as output a drawing of $G$ such that no two edges of the drawing
cross. Aesthetic requirements, such as minimizing the number of
bends along the edges, using a small area, or representing the faces
as convex or star-shaped polygons can be specified as additional
 optimization goals.
In addition to these geometric optimization goals, a graph drawing
algorithm can receive as input a set of semantic constraints which
define placement and/or routing requirements for some of the
vertices and/or some of the edges. For example, a subset of the
vertices must be collinear, some of the edges must be horizontal,
the size or shape of some edges or vertices is defined in advance
(see, e.g.,~\cite{dett-gd-99,kw-dg-01,s-gda-02}).

A well studied graph drawing problem with semantic constraints is
the  point set emebddability problem. Given a planar graph $G$ with
$n$ vertices and a set $S$ of $n$ distinct points in the plane, the
problem is to compute a crossing-free drawing of $G$ such that its
vertices are mapped to the points of $S$. Kaufmann and
Wiese~\cite{kw-evpfb-02}  proved that any planar graph has a
point-set embedding with at most two bends per edge on any given
point set. The drawing may however use exponential area, and a
polynomial area construction is showed
in~\cite{DBLP:conf/walcom/GiacomoL10}. If no bends along the edges
are allowed, Cabello~\cite{JGAA-132} proved that deciding whether a
planar graph admits a straight-line point-set embedding is an
NP-complete problem. Nishat et
al.~\cite{DBLP:journals/comgeo/NishatMR12} and Moosa and
Rahman~\cite{DBLP:journals/dmaa/MoosaR12} presented polynomial-time
algorithms to test straight-line point-set embeddabiliity of plane
3-trees. At last year SoCG, Biedl and
Vatshelle~\cite{DBLP:conf/compgeom/BiedlV12} showed that the problem
can be solved in polynomial time for planar graphs with a fixed
combinatorial embedding that have constant treewidth and constant
face-degree. They also show that if one of the conditions is dropped
(i.e., either the treewidth is unbounded or some faces have large
degrees), the problem becomes NP-hard. Gritzman, Mohar, Pach and
Pollack~\cite{gmpp-eptvs-91} and independently Casta\~neda and
Urrutia~\cite{cu-sepgps-96} proved that the outerplanar graphs with
$n$ vertices are the largest family of graphs admitting a
straight-line point-set embedding on any set of $n$ points in
general position.

Motivated by the above complexity results and restrictions on the
classes of planar graphs that admit a geometric point-set embedding,
we study the problem of drawing a planar graph with specified vertex
locations from a different angle. Namely, we consider a scenario in
which the given locations are fewer than the vertices of the graph.
In our setting the algorithm receives as input any planar graph $G$
with $n$ vertices and any ``sufficiently small'' set $S$ of points.
The output is a crossing-free straight-line drawing such that $|S|$
vertices of $G$ are mapped to the points of $S$, while the points
representing the remaining $n-|S|$ vertices of $G$ are chosen by the
algorithm. We call the output a {\em geometric point-subset
embedding of $G$ on $S$}. We stress that $S$ can be any set of
points, i.e. the drawing algorithm must work for all point sets of
``sufficiently small'' cardinality. We are interested in making the
``sufficiently small'' cardinality as large as possible. For
example, it is immediate to see that all planar graphs with at least
four vertices have a geometric point-subset embedding of $G$ on all
point sets of cardinality three. Less intuitively, it has been
proven in~\cite{DBLP:conf/isaac/Angelini12} that all planar graphs
with at least six vertices have a geometric point-subset embedding
of $G$ on all point sets of cardinality four. The following
significantly improves this cardinality.

\begin{theorem}\label{th:general}
Let $G$ be an $n$-vertex plane graph with vertices and let $S$ be a
set of distinct points in the plane such that $|S| \leq
\left \lceil \frac{\sqrt{\log_2 n}-1}{4} \right \rceil$. $G$ has a geometric point-subset embedding on $S$, which can be computed $O(n)$ time.
\end{theorem}

We also prove that if the points are in one-sided convex position
(i.e. they form a convex point set such that the leftmost and the
rightmost points are adjacent in the convex hull), the bound on the
cardinality in Theorem~\ref{th:general} can be further increased.


\begin{theorem}\label{th:convex}
Let $G$ be an $n$-vertex plane graph and let $S$  be a
one-sided convex point set in the plane such that $|S| \leq \left \lceil
\sqrt[3]{n} \right \rceil$.
$G$ has a geometric point-subset embedding
on $S$, which can be computed $O(n)$ time.
\end{theorem}

We recall that in ~\cite{DBLP:conf/isaac/Angelini12} it is proven
that there exists a one-sided convex point set $S$ of size $\left
\lceil \sqrt{n} \right \rceil$ such that any planar graphs admits a
geometric point-subset embedding on $S$. Theorem~\ref{th:convex} extends this
result to all one-sided convex point sets, even if with a slightly
smaller cardinality.

\paragraph{Overview of the proof technique.}
The proofs of Theorems~\ref{th:general} and~\ref{th:convex} follow a
similar approach, although they use different combinatorial results.
The general idea is described in the following steps. Let $G$ be an
$n$-vertex plane graph (i.e. a planar graph with a given planar
embedding) and let $S$ be a set of $k$ points with $k < n$.

\begin{description}

\item[Step~1:]Graph $G$ is augmented to a maximal plane graph $G'$ with $n$ vertices.

\item[Step~2:] An induced subgraph $H \subset G'$ of size at least $k$ is
computed.

\item[Step~3:] A point subset embedding $\Gamma_H$ of $H$ on $S$ is computed.

\item[Step~4:] A planar embedding preserving straight-line drawing of $G'$ is
constructed by adding vertices and edges to $\Gamma_H$.

\item[Step~5:]  The edges of $G' - G$ are removed.

\end{description}

Concerning Step~2, the induced subgraph $H$ is outerplanar in the
proofs of both theorems. In the proof of Theorem~\ref{th:convex},
$H$ is a biconnected outerplanar graph having at least $\left \lceil
\sqrt[3]{n} \right \rceil$ vertices; in the proof of
Theorem~\ref{th:general}, $H$ is an outerpath (i.e. a biconnected
outerplanar graph whose weak dual is a path) having at least $\left
\lceil \frac{\sqrt{\log_2 n}-1}{4} \right \rceil$ vertices. The
existence of these subgraphs will be proven in
Theorems~\ref{th:large-induced-outerpath} and
\ref{th:large-induced-outerplanar}, where it is also showed how to
compute them in $O(n)$ time. The arguments combine a variety of
concepts, including Schnyder's realizers~\cite{Schnyder90},
rectangular representations~\cite{Thomassen1986}, stabbing
techniques~\cite{Toth08}, 4-block trees~\cite{jgaa/WangHe12}, and
Mirsky's theorem~\cite{M-71}. As for Steps~3 and~4, they rely on
different drawing techniques concerning point-set embeddings and
straight-line drawings of plane graphs inside star-shaped
polygons~\cite{HongN08}.

We conclude this introduction by remarking that
Theorems~\ref{th:large-induced-outerpath} and
\ref{th:large-induced-outerplanar} may be of independent interest.
Namely, proving the existence of large induced subgraphs has a long
tradition in combinatorial graph theory. Erd{\"o}s, Saks, and
S{\'o}s~\cite{ErdosSS89} asked what is the maximum size of an
induced tree in a graph  and studied the relation between such size
and other natural graph parameters. For example, they give upper and
lower bounds on large induced trees in $K_r$-free graphs; these
bounds have been improved by Matou\v{s}ek and
\v{S}\'{a}mal~\cite{MatousekS07} and by Fox, Loh, and
Sudakov~\cite{FoxLS09}. Restricting to paths, Aroha and
Valencia~\cite{ArohaV00} showed that every 3-connected planar graph
with a large number of vertices has a long induced path. More
specifically, they proved that if a $3$-connected graph $G$ has a
vertex of degree $d$, then $G$ contains an induced path of length at
least $\sqrt{\log_3{d}}$. Generalizing this result to non-planar
graphs, B\"{o}hme et al.~\cite{BohmeMSS04} proved that for every
positive integers $k,~r$ and $s$ there exists an integer
$n=n(k,r,s)$ such that every $k$-connected graph of order at least
$n$ contains either an induced path of length $s$ or a subdivision
of the complete bipartite graph $K_{k,r}$. However, their
construction, do not seem to yield neither an explicit function for
the length of the path nor an algorithm to compute it. An
implication of our results is that every $n$ vertex maximal planar
graph has an induced path of length at least $\left \lceil
\frac{\sqrt{\log_2 n}-1}{4} \right \rceil$, which can be computed in
$O(n)$ time.


The rest of the paper is organized as follows. In
Section~\ref{se:preliminaries} we give preliminary definitions. The
techniques behind the proofs of Theorem~\ref{th:general}
and~\ref{th:convex} are described in Sections~\ref{se:general}
and~\ref{se:convex}, respectively. Finally, open problems can be
found in Section~\ref{se:conclusions}.

\section{Preliminaries}\label{se:preliminaries}

Let $G$ be a graph. A \emph{drawing} $\Gamma$ of $G$ is a mapping of
each vertex $v$ of $G$ to a point $p_v$ in the plane, and of each
edge $e=(u,v)$ to a Jordan arc connecting $p_u$ and $p_v$ not
passing through any other vertex. $\Gamma$ is \emph{straight-line}
if its edges are drawn as straight-line segments.  $\Gamma$ is
\emph{planar} if no two edges intersects (except possibly at common
endpoints). A planar drawing $\Gamma$ partitions the plane into
topologically connected regions called \emph{faces}; the unbounded
region is the \emph{external} face and the other regions are the
\emph{internal} faces. The boundary of a face is its delimiting
circuit. All the face boundaries of a biconnected graph are simple
circuits. A planar drawing determines a circular ordering of the
edges around each vertex. The cyclic ordering of the edges around each vertex of $\Gamma$
together with a choice of the external face is a \emph{planar
embedding} of $G$. A \emph{plane graph} is a graph with a fixed
planar embedding. The boundary of the external face of a plane graph $G$ will be also called the \emph{external boundary} of $G$. The \emph{dual graph} $G^\star$ of a given plane
graph $G$ is a plane graph that has a vertex corresponding to each
face of $G$, and an edge joining two vertices corresponding to
neighboring faces of $G$. The \emph{weak dual} of $G$ is the graph
obtained from $G^{\star}$ after removing the vertex that corresponds
to the external face of $G$. A $3$-cycle of a graph $G$ is also called a \emph{triangle} of $G$. A planar (plane) graph is
\emph{maximal} if each face of the graph is a triangle, thus no edge
can be added to it without violating planarity. A subgraph $H$ of a
graph $G$ is said to be \emph{induced} if, for any pair of vertices
$u$ and $v$ of $H$, $(u,v)$ is an edge of $H$ if and only if $(u,v)$
is an edge of $G$.

Let $G$ be a graph. A \emph{separating $k$-ple} of $G$ is a set of
$k$ vertices whose removal disconnects $G$.  A graph $G$ is
\emph{$k$-connected} if it does not contain a separating
$(k-1)$-ple. Every maximal planar graph is $3$-connected and every separating $3$-ple of $G$ (if any) is a triangle of $G$. Let $t_1$ and $t_2$ be two separating triangles of a maximal plane graph $G$; we say that $t_2$ is \emph{nested} inside $t_1$ and write $t_1 \succ t_2$ if at least one vertex of $t_2$ is inside $t_1$ in the planar embedding of $G$.

A graph is \emph{outerplanar} if it admits a planar embedding where
all vertices are on the external boundary. An \emph{outerpath} is an
outerplanar graph whose weak dual is a path. Let $H$ be an outerpath
and let $f_1$ and $f_2$ be the two faces corresponding to the
endpoints of $H^{\star}$. The edges of $f_1$ and $f_2$ that belong
to the external boundary of $H$ are called the \emph{extremal
edges} of $H$. Let $e_1$ be an extremal edge on the boundary of $f_1$, let $e_2$ be an extremal edge on the boundary of $f_2$. Any path in the external boundary of $H$ that connects an endvertex of $e_1$ with an endvertex of $e_2$ is called a \emph{side path} of $H$. It is easy to see that two extremal edges $e_1$ and $e_2$ define two side paths $\pi_1$ and $\pi_2$ and that every edge of $H$ that is not on the external boundary has an endvertex in $\pi_1$ and an endvertex in $\pi_2$. In other words both $\pi_1$ and $\pi_2$ are induced paths. Let $e$ be an extremal edge of $H$; any side path starting at an endvertex of $e$ is also called a \emph{side path of $H$ with respect to $e$}. Let $G$ be a maximal plane graph that contains an induced outerpath subgraph $H$ with an extremal edge $e$ on the external boundary of $G$ and such that there is a side path of $H$ with respect to $e$ that has $k$ vertices. We say that $G$ \emph{contains an external outerpath $H$ with a side path of size $k$}.


Let $S$ be a set of $k$ points on the plane, $k >0$ and let $G$ be a
planar graph with $n$ vertices, with $k \leq n$. A
\emph{geometric point-subset embedding} of $G$ on $S$ is a planar
straight-line drawing $\Gamma$ of $G$ such that each point of $S$ represents
a vertex of $G$. If $G$ is a plane graph and the planar embedding of $G$ defined by $\Gamma$ is the same as the one of $G$, we say that $\Gamma$ \emph{preserves} the planar embedding of $G$. All the algorithms of this paper compute drawings that preserve the planar embedding of the input graph. So when we write  ``geometric point-subset embedding'' we mean ``geometric point-subset embedding that preserves the planar embedding''.

\section{General Point Sets}\label{se:general}

In this section we prove Theorem~\ref{th:general}. The main ingredients of the proof are:
(i) A drawing algorithm that, given a set $S$ of $k$ points and a
maximal plane graph $G$ that contains an external outerpath with a side path of size $k$, computes a geometric point-subset embedding of $G$ on $S$
(Section~\ref{ss:general-embedding}). (ii) An algorithm that,
given a maximal plane graph $G$, computes an external outerpath of $G$ with a side path of size at least $\left \lceil
\frac{\sqrt{\log_2 n}-1}{4} \right \rceil$ (Section~\ref{ss:general-induced}).

\subsection{Drawing Tools}\label{ss:general-embedding}

Let $P$
be a polygon on the plane, a \emph{kernel} $K(P)$ of polygon $P$ is
the set of all points inside $P$ from which all vertices of $P$ are
visible. We say that $P$ is \emph{star-shaped} if $K(P)\neq
\emptyset$.


\begin{theorem}[Hong, Nagamochi~\cite{HongN08}]\label{theorem:HongNagamochi} Let $G$ be an $n$-vertex $3$-connected planar graph and let $C$ be its external boundary. Every planar straight-line drawing of $C$ as a star-shaped polygon can be extended to a planar straight-line drawing of $G$, where each internal face is represented by a convex polygon. Such a drawing can be computed in $O(n)$ time.
\end{theorem}

\begin{lemma}\label{le:embedding.1.a}
Let $S$ be a set $k$ distinct points in the plane and let $G$ be an
$n$-vertex maximal plane graph that contains an external outerpath with a side path of size at least $k$. $G$ admits a geometric point-subset embedding on $S$, which can be computed in $O(k \log k + n)$ time.
\end{lemma}
\begin{proof}
Assume that the points of $S$ have distinct $x$-coordinates. If not, we can rotate the plane so to achieve this condition. Let $p_1,p_2,\dots,p_k$ be the points of $S$ according to their left-to-right order and let $(x_i,y_i)$ be the coordinates of $p_i$ ($i=1,2,\dots,k$). Without loss of generality assume $x_1=0$, $y_1=0$. If this is not the case we can translate the plane to achieve this condition. Let $P$ be the side path of the external outerpath $H$ of $G$ that has size at least $k$. If $P$ has more than $k$ vertices we can augment $S$ with additional points. So without loss of generality assume that $P$ has exactly $k$ vertices. Let $u_1, u_2, \dots, u_k$ be the vertices of $P$ in the order they appear along $P$ with $u_1$ being the vertex of $P$ on the external boundary of $G$. Map $u_i$ to $p_i$ and draw each edge $(u_i,u_{i+1})$ as the segment $\seg{p_ip_{i+1}}$ (see Figure~\ref{fi:drawing1} for an illustration of the drawing technique). All vertices of $H$ that are not in $P$ form a second side path $P'$ of $H$. Since $P'$ is also an induced path, all the edges of $H$ that are not in $P$ nor in $P'$ connect vertices of $P$ to vertices of $P'$. Let $v_1,v_2,\dots,v_h$ be the vertices of $P'$ in the order they appear along $P'$ with $v_h$ being the vertex of $P'$ on the external boundary of $G$. We want to place these vertices below the points of $S$ in such a way that each point of $S$ can be connected to the vertices of $P'$ without crossing the segments $\seg{p_ip_{i+1}}$. Let $\delta_x= \min_i \{|x_{i+1}-x_i|\}$, let $\Delta_y= \max_{i,j} \{|y_i-y_j|\}$, and let $\sigma=(\Delta_y+1)/\delta_x$. Let $\sigma_i$ be the slope of segment $\seg{p_ip_{i+1}}$ ($i=1,2,\dots,k-1$). We have $-\sigma < \sigma_i < \sigma$ for every $i=1,2,\dots,k-1$.  We place vertex $v_i$ at point $q_i=(x_k+1,-\sigma (x_k+i))$. The slope of each segment $\seg{p_iq_j}$ is at most $-\sigma$ and therefore each $\seg{p_iq_j}$ does not cross any segment $\seg{p_ip_{i+1}}$. Also, segments $\seg{p_iq_j}$ do not cross each other. Namely, let $\seg{p_{i_1}q_{j_1}}$ and $\seg{p_{i_2}q_{j_2}}$ be two such segments. Since the left-to-right order of the vertices of $P$ in the drawing is the order in which they appear along $P$ going from the external boundary to the interior, and the bottom-to-top order of the vertices of $P'$ in the drawing is the order in which they appear along $P'$ going from the external boundary to the interior, then $i_1 \leq i_2$ and $j_2 \leq j_1$. This means that the slope of $\seg{p_{i_1}q_{j_1}}$ is less than the slope of $\seg{p_{i_2}q_{j_2}}$; since $q_{j_1}$ is below $q_{j_2}$, the two segments do not cross each other. The drawing computed so far is a geometric point-subset embedding of $H$ on $S$. Also, the drawing preserves the embedding of $H$. Let $w$ be the unique vertex of the external boundary that does not belong to $H$. In order to complete the drawing we place the vertex $w$ at point $(x_k+2,\sigma(x_k+2))$ and draw all the edges connecting it to all its adjacent vertices that have already been drawn. Also in this case, the added edges do not cross the existing ones. So the obtained drawing $\Gamma'$ is a geometric point-subset embedding of the graph $H'$ induced by the vertices in $H$ plus vertex $w$. and also in this case the planar embedding is preserved. It is also easy to see that all the faces of $\Gamma'$ are star-shaped polygons. Notice that, in the planar embedding of $G$ each vertex not yet drawn is inside one of the cycles represented by the faces of $H'$. Let $C$ be one of these cycles and let $G'$ be the graph induced by the vertices of $C$ plus all the vertices that are inside $C$ in the planar embedding of $G$. All the internal faces of $G'$ are triangles, while its external boundary is $C$ (which therefore can have more than three vertices). Notice that there is no edge connecting two non-consecutive vertices of the external boundary of $G'$; namely if such an edge existed it would have been drawn in $\Gamma'$ and therefore $C$ would not have been a cycle of $H'$. An internally triangulated graph is $3$-connected if and only if there is no edge connecting two non-consecutive vertices of its external boundary (see, for example,~\cite{Avis96}). Hence, $G'$ is $3$-connected and its external boundary is drawn as a star-shaped polygon. By Theorem~\ref{theorem:HongNagamochi} the drawing of $C$ can be extended to a planar drawing of $G'$. Since this is true for all cycles that contain some not yet drawn vertices, the drawing $\Gamma'$ can be completed to a geometric point-subset embedding $\Gamma$ of $G$ on $S$. Since a $3$-connected planar graph has a unique embedding once the external face is fixed, the completion of $\Gamma'$ to $\Gamma$ preserves the embedding of $G$.

Concerning the time complexity, it is easy to see that the drawing can be computed in $O(n)$ time once the input points are sorted from left to right. Such a sorting preprocessing requires time $O(k \log k)$. \end{proof}

\begin{figure}
\centering
\subfigure[]{\label{fi:drawing1}\includegraphics[scale=0.5]{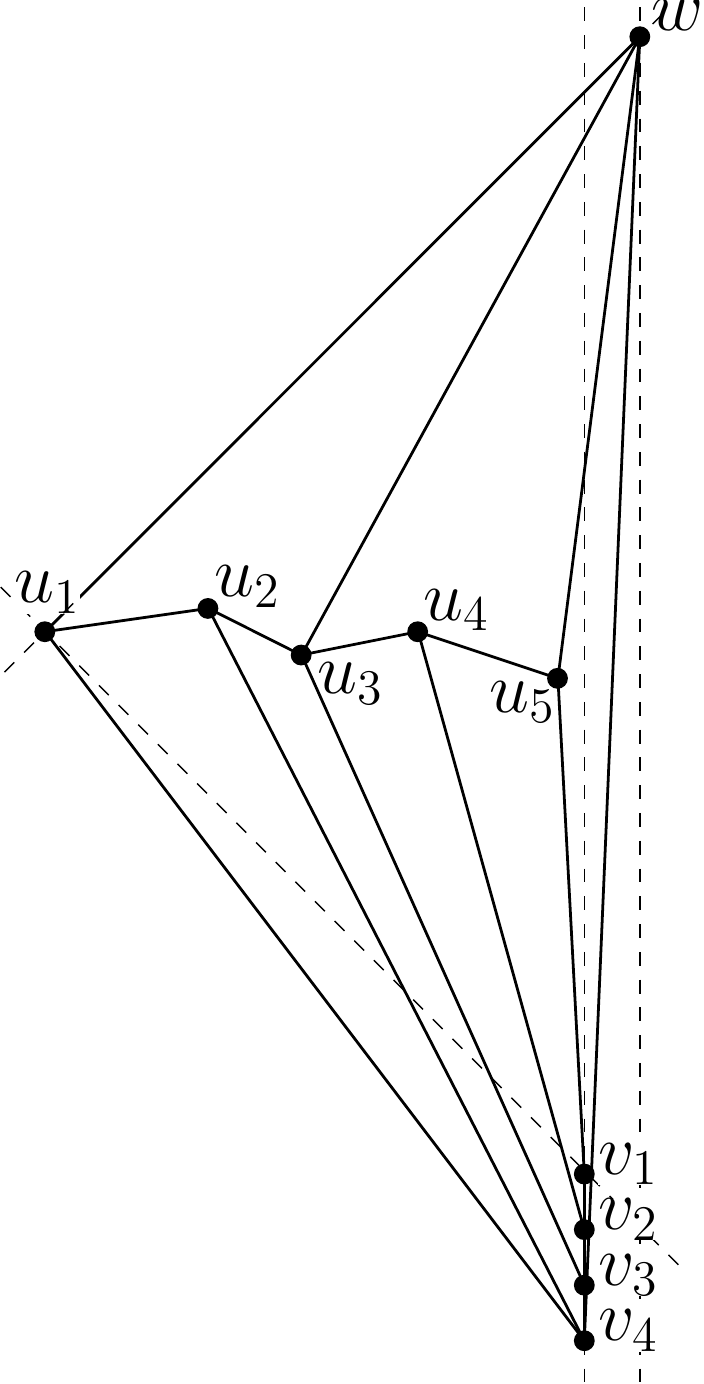}}
\hfil
\subfigure[]{\label{fi:drawing2}\includegraphics[scale=0.5]{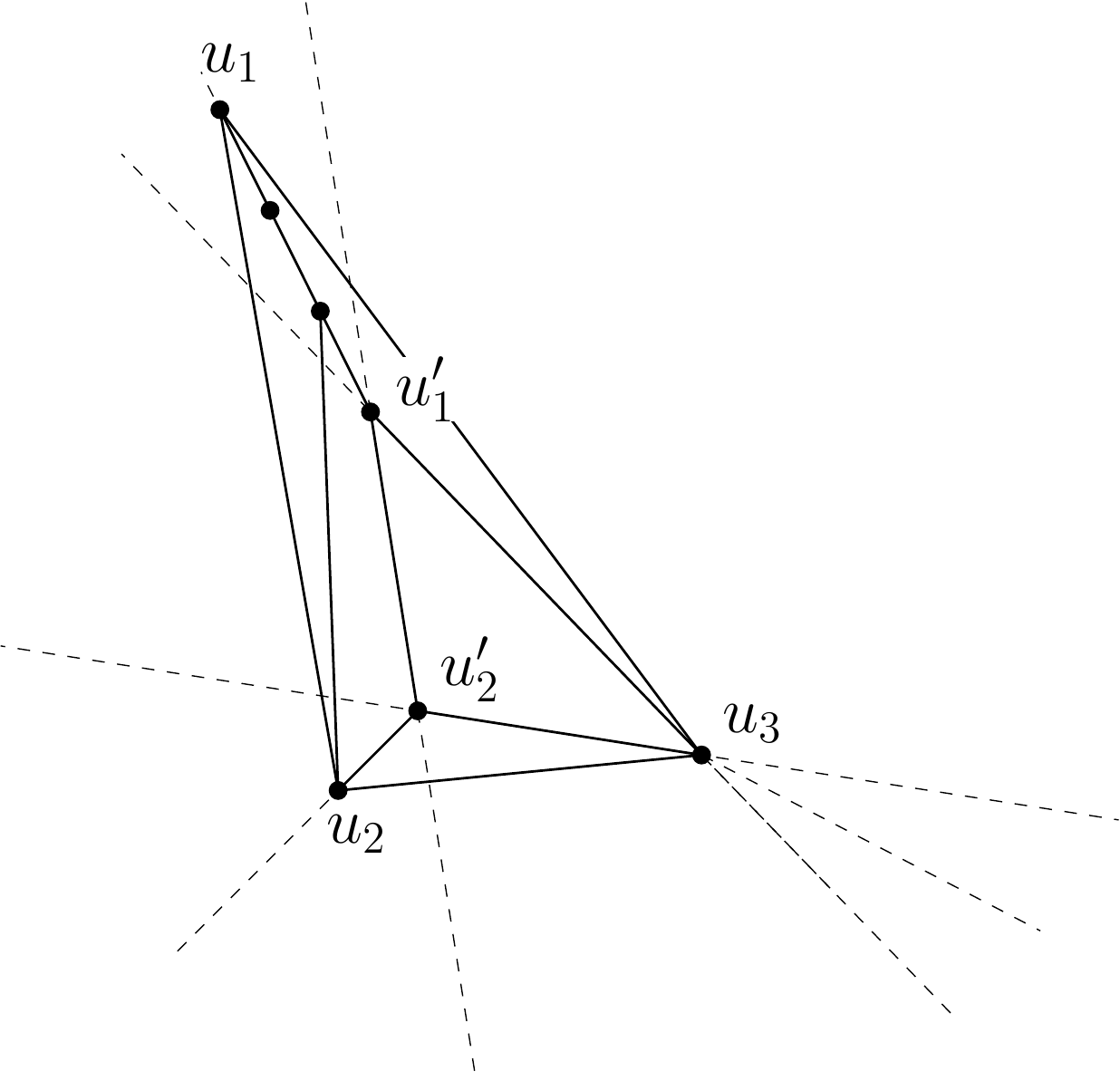}}
\caption{(a) Illustration of the drawing technique of Lemma~\ref{le:embedding.1.a}. (b) Illustration of the drawing technique of Lemma~\ref{le:embedding.1.b}.}
\end{figure}

\begin{lemma}\label{le:embedding.1.b}
Let $G$ be an $n$-vertex maximal plane graph. Let $u_1$, $u_2$, and $u_3$ be the three external vertices of $G$ and let $u'_1$, $u'_2$, and $u'_3$ be three vertices of an internal face $f$ such that there exists three induced paths $u_1-u'_1$, $u_2-u'_2$, and $u_3-u'_3$ in $G$ (possibly of zero length). Let $\tau$ be a straight-line drawing of $f$ on any set $S$ of three points. $G$ admits a geometric point-subset embedding on $S$ that has $\tau$ as a subdrawing. Such a point-subset embedding can be computed in $O(n)$ time.
\end{lemma}
\begin{proof}
Let $p_1$, $p_2$, and $p_3$ be the three points representing  $u_1$, $u_2$, and $u_3$, respectively. Let $\rho_{i}$ be the bisector of the angle at $p_i$ of triangle $\tau$ ($0 \leq i \leq 2$), and let $\rho^+_i$ be the half-line contained in $\rho_i$ that is completely outside $\tau$ (except for point $p_i$). Let $u_i=v_{i,0}, v_{i,1}, \dots, v_{i,k_i}=u'_i$ be the vertices of the path $u_i-u'_i$ ($0 \leq i \leq 2$). Place vertex $v_{i,j}$ along $\rho^+_i$ at distance $j$ from $p_i$ (see Figure~\ref{fi:drawing2} for an illustration of the drawing technique). Since each of the paths $u_i-u'_i$ is induced, there is no edge connecting non consecutive vertices along the path, and therefore each path can be drawn as a subdivision of a segment. We also add to the drawing all segments representing edges connecting the vertices of the three induced paths. Since the angle between $\rho_i$  and $\rho^+_i$ is less than $\pi$, and since the vertices of each path $u_i--u'_i$ are placed along $\rho^+_i$ in the same order they have along the path $u_i--u'_i$, the resulting drawing is planar. Moreover, each face of the drawing is drawn as a star-shaped polygon. In the planar embedding of $G$ each vertex not yet drawn is inside one of the cycles represented by the faces of the computed drawing. Thus, the drawing can be completed by using Theorem~\ref{theorem:HongNagamochi} as described in the proof of Lemma~\ref{le:embedding.1.a}. It is immediate to see that the computed drawing preserves the embedding and that the algorithm time complexity in $O(n)$.
\end{proof}

\subsection{Large Induced Outerpath}\label{ss:general-induced}

We now show that each $n$-vertex maximal plane graph
contains an external outerpath with a side path of size $\left \lceil
\frac{\sqrt{\log_2{n}}-1}{4}\right \rceil$. Our proof is based on the fact that an $n$-vertex maximal plane graph contains either many nested separating triangles or a large $4$-connected subgraph. Thus, we first concentrate on
graphs with many separating triangles and then on $4$-connected graphs.

\subsubsection{Large Induced Outerpath in Graphs with Many Nested Triangles}

In order to prove that a graph with many separating triangles contains a large induced outerpath, we need to introduce the Schnyder realizer~\cite{Schnyder90}, which will also be used in Section~\ref{se:convex}.

\paragraph{Schnyder realizers.} Let $G$ be a maximal plane graph with $n \geq 3$ vertices. A \emph{Schnyder realizer}~\cite{Schnyder90} of $G$ is an orientation plus a coloring of the internal edges of $G$ such that:
\begin{enumerate}
\item Each internal edge receives one of three colors $0$, $1$, and $2$;
\item Each internal vertex has exactly one outgoing edge of each color;
\item The counterclockwise order of the outgoing edges around each internal vertex is $0$, $1$, $2$;
\item For each internal vertex, the incoming edges of a color appear counterclockwise between the two outgoing edges of the other two colors.
\end{enumerate}

\begin{figure}
\centering
\subfigure[]{\includegraphics[scale=0.95]{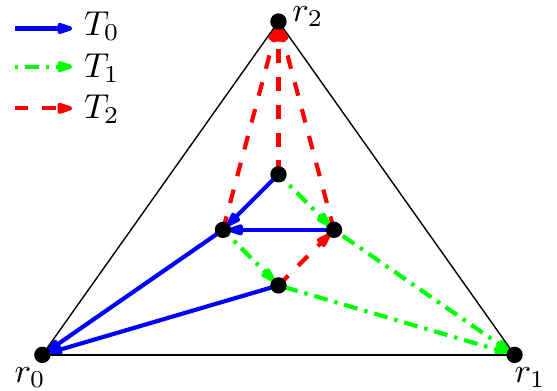}\label{fi:schnyder}}
\subfigure[]{\includegraphics[scale=0.95]{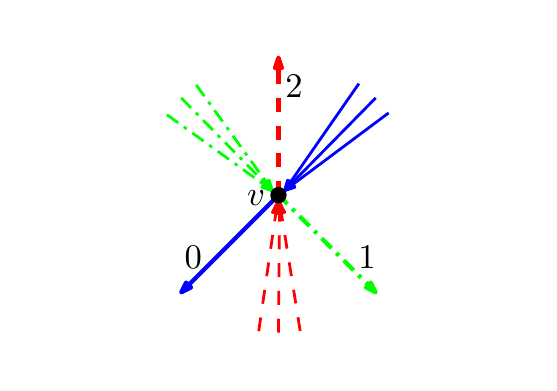}\label{fi:circularOrder}}
\subfigure[]{\includegraphics[scale=0.95]{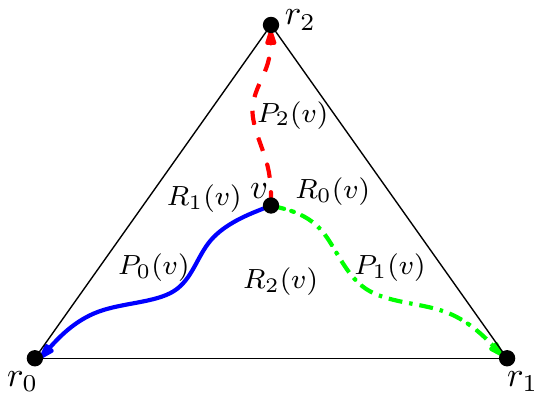}\label{fi:regions}}
\caption{(a) A Schnyder realizer $\T=\{T_1,T_2,T_3\}$ of a maximal plane graph. (b) The circular ordering of the edges around an internal vertex $v$. (c) The three paths $P_1(v)$, $P_2(v)$, and $P_3(v)$ and the three regions $R_1(v)$, $R_2(v)$, and $R_3(v)$}
\end{figure}

See Figure~\ref{fi:schnyder} for an example of a Schnyder realizer of a maximal plane graph equipped with. Note that conditions $3$ and $4$ of the definition of a Schnyder realizer imply that the counterclockwise circular ordering of the edges around each internal vertex is the one shown in Figure~\ref{fi:circularOrder}, i.e., outgoing colored $0$, incoming colored $2$, outgoing colored $1$, incoming colored $0$ outgoing colored $2$, incoming colored $1$, where each of the incoming set can be empty.
It is known that every maximal planar embedded graph has a Schnyder realizer. Let $T_0$, $T_1$, and $T_2$ be the subgraph induced by the edges of color $0$, $1$, and $2$, respectively. It is possible to prove each of $T_0$, $T_1$, and $T_2$ is a tree containing all inner vertices and exactly one outer vertex. Each edge in each tree is oriented from a vertex to its parent and the single outer vertex of each tree has only incoming edges and therefore is the root of the tree. We will denote the root of each tree $T_i$ by $r_i$, $0 \leq i \leq 2$. The three roots $r_0$, $r_1$, and $r_2$ are distinct and their counterclockwise order on the external boundary is $0$, $1$, and $2$. In the following we denote a Schnyder realizer $\{T_0, T_1, T_2\}$ of $G$ as $\T$.

Let $G$ be a maximal plane graph equipped with a Schnyder realizer $\T$. Let $v$ be an internal vertex and let $P_i(v)$, $0 \leq i \leq 2$ be the oriented path of $T_i$ from $v$ to $r_i$. The path $P_i(v)$ is called the \textit{$i$-path starting at $v$}. For $i \neq j$ ($0 \leq i,j \leq 2$), the two paths $P_i(v)$ and $P_j(v)$ only share the vertex $v$. Thus, for each internal vertex $v$, the three paths $P_0(v)$, $P_1(v)$, and $P_2(v)$ divide $G$ into three regions $R_0(v)$, $R_1(v)$, and $R_2(v)$, where $R_i(v)$ ($0 \leq i \leq 2$) denotes the vertices that are inside the cycle $P_{i-1}(v) \cup P_{i+1}(v) \cup (r_{i-1},r_{i+1})$ (where indices are taken modulo $3$). Notice that the regions $R_0(v)$, $R_1(v)$, and $R_2(v)$ do not contain the vertices of $P_0(v)$, $P_1(v)$, and $P_2(v)$\footnote{In the literature these regions are usually defined as including also $P_0(v)$, $P_1(v)$, and $P_2(v)$. For our purposes, however, it results more useful to use this different definition.}. The following lemmas can be proven by exploiting the properties of Schnyder realizer.

\begin{lemma}\label{lemma:Schnyder_induced}
Let $G$ be a maximal plane graph equipped with a Schnyder realizer $\T$ and let $v$ be an internal vertex of $G$. Each $P_i(v)$ ($0 \leq i \leq 2$) is an induced path of $G$.
\end{lemma}
\begin{proof}
Let $v=u_0, u_1, \dots, u_k=r_i$ be the vertices of $P_i(v)$ in the order they appear walking along $P_i(v)$ from $v$ to $r_i$. Suppose, as a contradiction that $P_i(v)$ is not induced. This means that $P_i(v)$ has a chord, i.e., an edge connecting two non consecutive vertices of $P_i(v)$. We first observe that the chord must be completely contained either in the region $R_{i+1}(v)$ or in the region $R_{i-1}(v)$ (except for the endvertices), because otherwise either $r_i$ would not be on the external boundary, or the chord would be crossed by $P_{i+1}(v)$ and by $P_{i-1}(v)$. Also, the chord cannot be an edge colored $i$, otherwise there would be a cycle in $T_i$. As a consequence the chord cannot be incident to vertex $r_i$ (all the edges incident to $r_i$ have color $i$ except the two on the external boundary). For each vertex $u_j$ ($0 < j < k$) the edge $(u_{j},u_{j+1})$ is the unique outgoing edge of $u_j$ colored $i$, while $(u_{j-1},u_{j})$ is an incoming edge of $u_j$ colored $i$. By the properties of the Schnyder realizer the outgoing edge of $u_j$ colored $i+1$ appears between $(u_{j},u_{j+1})$ and $(u_{j-1},u_{j})$ in the counterclockwise order around $u_j$, while the incoming edges colored $i+1$ appears between  $(u_{j-1},u_{j})$ and $(u_{j},u_{j+1})$ in the counterclockwise order around $u_j$. In other words, the outgoing edge of $u_j$ colored $i+1$ and the incoming edges of $u_j$ colored $i+1$ are on different sides of the path $P_i(v)$. The same is true also for the edges colored $i-1$. This implies that there cannot be a chord colored $i+1$ or $i-1$ connecting two vertices $u_j$ and $u_l$ with $0 < j,l < k$. The only possibility remaining is that the chord connects $v$ to a vertex $u_j$ with $0 < j < k$ and, as already said, it is colored either $i+1$ or $i-1$. The chord cannot be an outgoing edge of $v$ otherwise one of the two paths $P_{i+1}(v)$ and $P_{i-1}(v)$ would share a vertex with $P_i(v)$, which is not possible. So the chord $e=(v,u_j)$ must be an incoming edge of $v$ colored $i+1$ or $i-1$. Suppose it is colored $i+1$ (the other case is symmetric). Then $e$ must be contained in the region $R_{i+1}(v)$ (it must appear counterclockwise between the outgoing edge of $v$ colored $i$ and the outgoing edge of $v$ colored $i-1$). On the other hand, $e$ is also an outgoing edge of $u_j$ colored $i+1$ and therefore it must be contained in the region $R_{i-1}(v)$ (it must appear counterclockwise between $(u_{j},u_{j+1})$ and $(u_{j-1},u_{j})$). But, as already said, a chord must be contained in only one region, and therefore we have a contradiction.\end{proof}

\begin{lemma}\label{lemma:Schnyder_outerpath}
Let $G$ be a maximal plane graph equipped with a Schnyder realizer $\T$ and let $v$ be an internal vertex of $G$. The subgraph $H_{ij}$ of $G$ induced by $P_i(v)$, $P_j(v)$ ($0 \leq i,j \leq 2$) is an outerpath with the edge $(r_i,r_j)$ on the external boundary of $G$. Also, $P_i(v)$ and $P_j(v)$ are both side paths of $H_{ij}$ with respect to $(r_i,r_j)$.
\end{lemma}
\begin{proof}
Let $C$ be the cycle $P_i(v) \cup P_j(v) \cup (r_i,r_j)$. Since $P_i(v)$ and $P_j(v)$ are both induced paths, then every edge of $H_{ij}$ that is not in $C$ connects a vertex of $P_i(v)$ to a vertex of $P_j(v)$. If one of such edges was outside $C$, then it would be crossed by path $P_k(v)$ ($k \neq i, k\neq j$). Since all edges not in $C$ are inside $C$, then $H_{ij}$ is outerplanar. Since each edge not in $C$ connects a vertex of $P_i(v)$ to a vertex of $P_j(v)$, then $H_{ij}$ is an outerpath and both $P_i(v)$ and $P_j(v)$ are its side paths. Clearly, edge $(r_i,r_j)$ is on the external boundary of $G$ and both, $P_i(v)$ and $P_j(v)$ are side paths with respect to $(r_i,r_j)$.
\end{proof}


A \emph{sequence of nested separating triangles} of a maximal plane graph $G$ is a set of $k$ separating triangles $t_1, t_2, \dots, t_k$ such that $t_{i} \succ t_{i+1}$, for $i=1,2,\dots,k-1$.

\begin{lemma}\label{lemma:ordered_triangles} Let $G$ be an $n$-vertex maximal plane graph. If $G$ contains a sequence of $k$ nested
separating triangles, then $G$ contains an external outerpath with a side path of size at least $\left \lceil (k+2)/3 \right \rceil$.
\end{lemma}
\begin{proof}
Let $t_1 \succ t_2 \succ \dots \succ t_k$ be the sequence of nested separating triangles of $G$ (see Figure~\ref{fig:sep_trianges_2}).
Since $t_k$ is a separating triangle, there exists a vertex $u$
which lies inside $t_k$. Let $\T$ be a Schnyder realizer of
$G$. Since the three paths $P_0(u)$, $P_1(u)$, and $P_2(u)$ of $G$ are disjoint, each path
visits each triangle $t_i$ ($i=1,2,\dots,k$) exactly once. Triangles $t_1,\dots,t_k$ contain in total at least
$k+2$ distinct vertices. Therefore at least one among $P_0(u)$, $P_1(u)$, and $P_2(u)$ contains at least $\left \lceil (k+2)/3 \right \rceil$ vertices.
Suppose it is $P_0(u)$. By Lemma~\ref{lemma:Schnyder_outerpath}, the graph $H_{01}$ induced by the vertices of $P_0(u)$ and $P_1(u)$ is an external outerpath of $G$ with a side path of size at least $\left \lceil (k+2)/3 \right \rceil$ (the same is true for the graph $H_{02}$ induced by the vertices of $P_0(u)$ and $P_2(u)$).
\end{proof}

\begin{figure}[htbp]
 \centering
 \subfigure[\label{fig:sep_trianges_2}]
 {\includegraphics[scale=0.95]{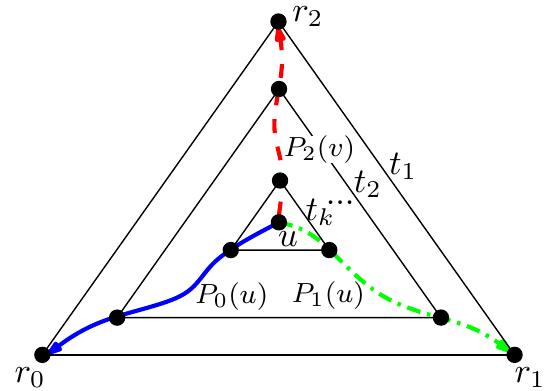}}
 \hfil
 \subfigure[\label{fig:thomassenconf}]
 {\includegraphics[scale=0.6]{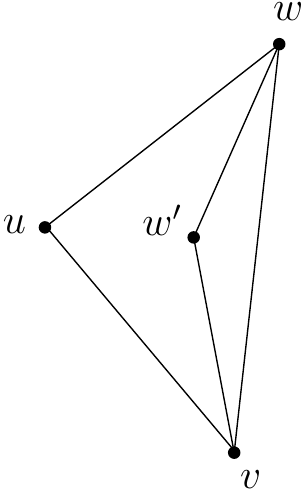}}
 \hfil
 \subfigure[\label{fig:thomassen}]
 {\includegraphics[scale=0.6]{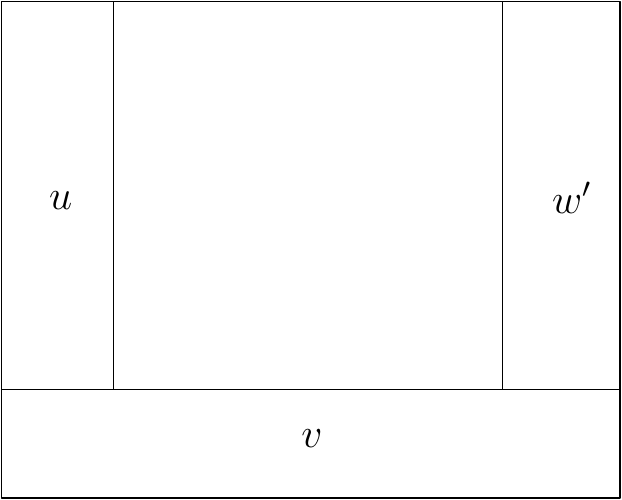}}
 \caption{(a) Illustration for the proof of Lemma~\ref{lemma:ordered_triangles}. (b)-(c) Illustration of the Theorem~\ref{theorem:thomassen}.}
 \label{fig:prelim}
\end{figure}

\subsubsection{Large Induced Outerpath in $4$-connected Graphs}

The proof that a $4$-connected graph contains a large induced outerpath relies on the following result by
Thomassen~\cite{Thomassen1986}. A planar drawing of $G$ such that
all edges (except four) are vertical or horizontal straight-line
segments and such that the four exceptional edges each consist of a
vertical and a horizontal straight-line segment is called
\emph{rectangular representation} of $G$.

\begin{theorem}[Thomassen~\cite{Thomassen1986}]\label{theorem:thomassen}
Let $G$ be a $4$-connected maximal plane graph and let $uvw$ and
$wvw'$ be two triangles of $G$. Then $G^\star$ has a rectangular
representation such that $w$ corresponds to the unbounded face and
such that $u$, $v$, and $w'$ correspond to rectangles as shown in
Figure~\ref{fig:thomassen}.
\end{theorem}

\begin{lemma}\label{lemma:four_connected}
Let $G$ be an $n$-vertex $4$-connected maximal plane graph.
Then $G$ contains an external outerpath with a side path of size at least
$\left \lceil \log_2{(n-1)}/4 \right \rceil$.
\end{lemma}
\begin{proof}
Let $u,w,v$ be the external boundary of $G$ and let $w'$ be the vertex of $G$, such that $w,v,w'$ is a face of $G$. By Theorem~\ref{theorem:thomassen}, $G^\star$ has a rectangular representation such that $w$ corresponds to the unbounded face and
such that $u$, $v$, and $w'$ correspond to rectangles as shown in Figures~\ref{fig:thomassen} and~\ref{fig:thomassenconf}. I.e, $G$ is the dual of a subdivision of a
rectangle into $n-1$ rectangles. In~\cite{Toth08} it is shown that there exists a horizontal or vertical line (called a \emph{stabber}) that
intersects at least $\left \lceil \log_2{(n-1)}/4 \right \rceil$ rectangles in this subdivision. The rectangles stabbed by a stabber represent an induced path $P$ of $G$. 

\begin{figure}[htbp]
 \centering
 \subfigure[\label{fig:thomassen_horizontal}]
 {\includegraphics[scale=0.6]{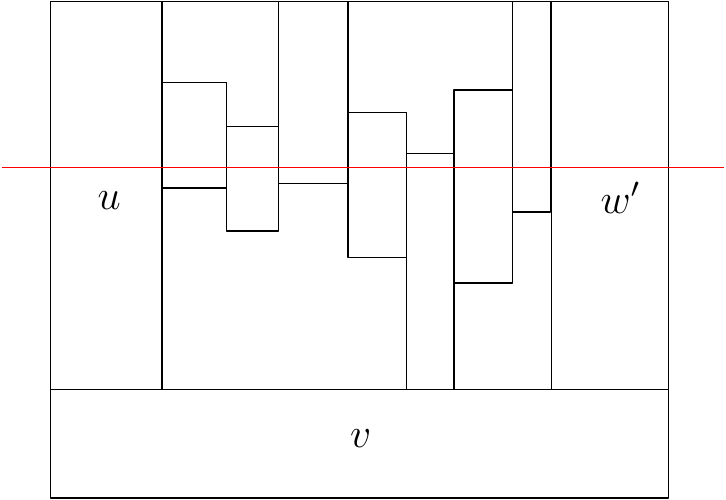}}
 \hfil
 \subfigure[\label{fig:thomassenconf_horizontal}]
 {\includegraphics[scale=0.6]{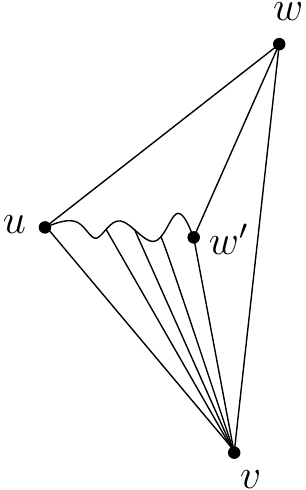}}
 \hfil
 \subfigure[\label{fig:thomassen_vertical}]
 {\includegraphics[scale=0.6]{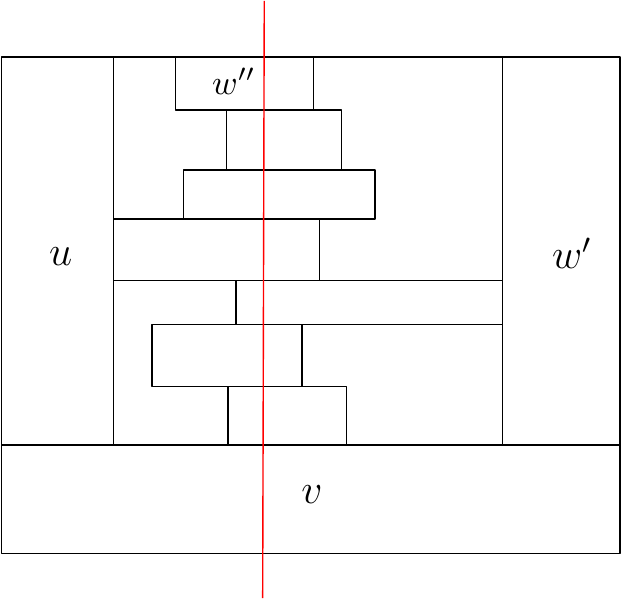}}
 \hfil
 \subfigure[\label{fig:thomassenconf_vertical}]
 {\includegraphics[scale=0.6]{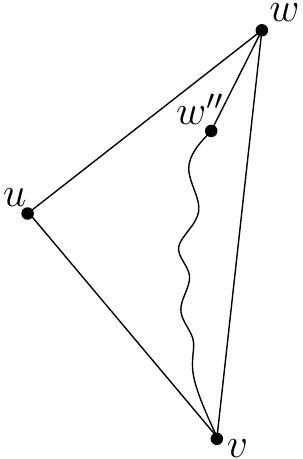}}
 \caption{Illustration of the proof of Lemma~\ref{lemma:four_connected}. (a)-(b) Long stabber is horizontal. (c)-(d) Long stabber is vertical.}
 \label{fig:horizontal}
\end{figure}

In the following we show that $G$ contains an external outerpath that has $P$ as a side path.
We consider two cases based on whether the long stabber is vertical or horizontal.

\begin{description}
\item[Long stabber is horizontal] Refer to Figures~\ref{fig:thomassen_horizontal}-~\ref{fig:thomassenconf_horizontal}. Path $P$ connects $u$ to $w'$ and edges $(w',v)$ and $(u,v)$ belong to $G$. Let $C$ be the cycle $P \cup (w',v) \cup (u,v)$ and let $H$ be the subgraph induced by the vertices in $C$. Since $P$ is an induced path, each edge of $H$ that does not belong to $C$ is incident to $v$. From the reactangular representation, it is immediate to see that all the edges connecting $v$ to $P$ lie between $(v,u)$ and $(v,w')$ in the clockwise circular order around $v$, i.e., they are inside the cycle $C$. Thus, $H$ is an outerpath; it has an extremal edge (i.e., the edge $(u,v)$) on the external boundary of $G$; and $P$ is a side path of $H$ with respect to $(u,v)$.

\item[Long stabber is vertical] Refer to Figures~\ref{fig:thomassen_vertical}-~\ref{fig:thomassenconf_vertical}. Let $w''$ be the vertex representing the highest rectangle stabbed by the stabber. Edge $(w'',w)$ belongs to $G$ (recall that $w$ is represented by the unbounded face). Moreover, vertex $w$ is not
adjacent to any other vertex of $P$ except for $v$. Let $C$ be the cycle $P \cup (w'',w) \cup (w,v)$. Since $P$ is an induced path and there is no edge connecting $w$ to a vertex of $P$ distinct from $v$, the subgraph $H$ induced by the vertices of $C$ coincides with $C$ itself, i.e., $H$ is a cycle. Thus, $H$ is an outerpath; it has an extremal edge (i.e., the edge $(v,w)$) on the external boundary of $G$; and $P$ is a side path of $H$ with respect to $(u,v)$.
\end{description}
\end{proof}

\subsection{Putting All Together}

We are now ready to combine the results of Section~\ref{ss:general-embedding} and~\ref{ss:general-induced} in order to prove Theorem~\ref{th:general}. To this aim, we use the following decomposition tree.

\paragraph{4-block tree.} Let $G$ be a maximal plane graph.  A \emph{$4$-block tree}~\cite{jgaa/WangHe12}
$T$ of $G$ is a directed tree containing a vertex for each triangle
of $G$ (refer to Figure~\ref{fig:4block}.a-b). Triangle of $G$ that
corresponds to a vertex $\mu$ of $T$ is denoted by
$\triangle_{\mu}$. $T$ contains edge directed from its vertex $\mu$
to its vertex $\nu$ if $\triangle_{\nu} \prec \triangle_{\mu} $ and
there is no triangle $\triangle$ in $G$ such that $\triangle_{\nu} \prec \triangle \prec \triangle_{\mu}$. We assume that $T$ is rooted
at its source i.e. at the vertex which corresponds to the external
face of $G$. Note that the leaves of $T$ correspond to the internal
faces of $G$. The subgraph of $G$ induced by the vertices of
$\triangle_{\mu}$ and all the vertices of the children of $\mu$ is
called the \emph{pertinent} graph of $\mu$ and is denoted by
$G_{\mu}$ (see Figure~\ref{fig:4blocktree_pert}).

\begin{figure}[htbp]
 \centering
 \subfigure[\label{fig:4blocktree}]
 {\includegraphics[scale=0.6]{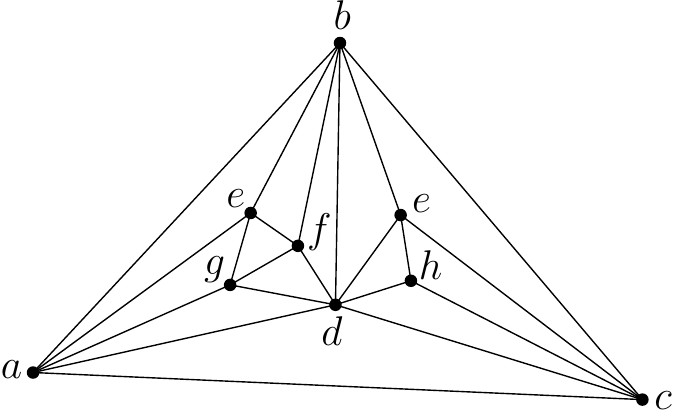}}
 \hfil
 \subfigure[\label{fig:4blocktree_decomp}]
 {\includegraphics[scale=0.6]{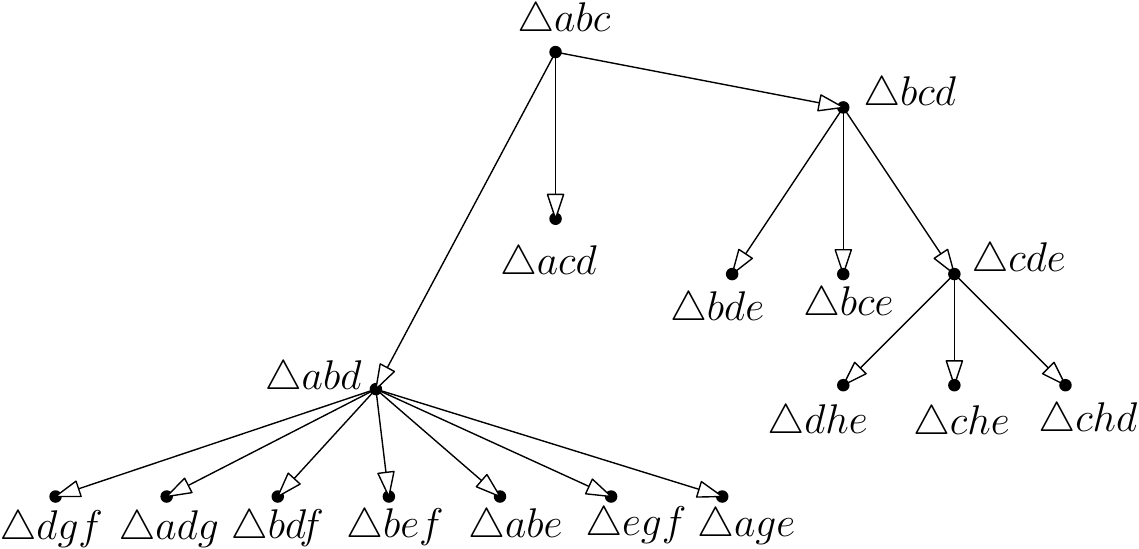}}
 \hfil
 \subfigure[\label{fig:4blocktree_pert}]
 {\includegraphics[scale=0.6]{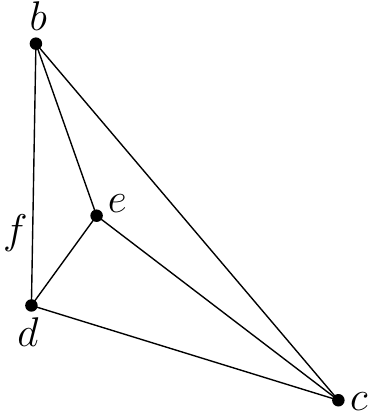}}
 \caption{(a) A plane graph $G$. (b) An extended $4$-block tree of $G$. (c)  Pertinent graph $G_{\mu}$ of node $\mu=\triangle bcd$}
 \label{fig:4block}
\end{figure}

\begin{theorem}\label{th:large-induced-outerpath}
In every $n$-vertex maximal plane graph $G$ there exists a maximal plane subgraph $G'$ that contains an external outerpath $H$ with a side path of size at least $ \left \lceil \frac{\sqrt{\log_2 n}-1}{2}  \right \rceil$. Both $G'$ and $H$ can be computed in $O(n)$-time.
\end{theorem}
\begin{proof}
Let $G$ be a $n$-vertex maximal plane graph and let $T$ be its
triangle decomposition tree. Let $\alpha(n)$ be a function of $n$ to
be specified later. We consider two cases based on whether the
degree of $T$ is bounded by $\alpha(n)$ or not.

\begin{description}
\item{\bf The maximum outdegree of $T$ is at most $\alpha(n)$.}
Since $G$ is a maximal plane graph it contains exactly $2n-4$ faces.
Thus, $T$ contains $2n-4$ leaves and therefore, the height of $T$ is
at least $\left \lceil  \log_{\alpha(n)}{(2n-4)} \right \rceil$. This means that $G$ contains a sequence of at least $k=\left \lceil \log_{\alpha(n)}{(2n-4)} \right \rceil+1$ nested separating triangles. By  Lemma~\ref{lemma:ordered_triangles}, $G$ contains an external outerpath $H$ with a side path of size at least $\left \lceil \frac{k+2}{3}\right \rceil =  \left \lceil \frac{\left \lceil \log_{\alpha(n)}{(2n-4)} \right \rceil +3}{3}  \right \rceil$.
\item{\bf There exists a vertex of $T$ with outdegree at least $\alpha(n)+1$.}
Let $\mu$ be such a vertex. Then the pertinent graph $G_{\mu}$ of
$\mu$ contains at least $\alpha(n)+1$ faces. Since $G_{\mu}$ is
maximal planar graph, it contains at least $n_{\mu}=\left \lceil \frac{\alpha(n)+1}{2}
\right \rceil +2 $ vertices. By  Lemma~\ref{lemma:four_connected}, $G_{\mu}$ contains an external outerpath $H$ with a side path of size at least $\left \lceil \frac{\log_2{(n_{\mu}-1)}}{4} \right \rceil = \left \lceil \frac{\log_2{(\left \lceil (\alpha(n)+1)/2 \right \rceil +1)}}{4} \right \rceil$.
\end{description}

\noindent Thus, the size of the side path of $H$ is
$$
h=\min \left\{
 \left \lceil \frac{\lceil \log_{\alpha(n)}{(2n-4)} \rceil +3}{3}  \right \rceil,
 \left \lceil \frac{\log_2{( \lceil (\alpha(n)+1)/2  \rceil +1)}}{4}  \right \rceil
\right \},
$$

By setting $\alpha(n)=2^{\sqrt{\lg{n}}}$, we get that
$h \geq \min \left \{
\left \lceil \frac{ \lceil \sqrt{\log_2 n}+3  \rceil}{3} \right \rceil,
\left \lceil \frac{\sqrt{\log_2{n}-1}}{4}\right \rceil \right \}
= \left \lceil \frac{\sqrt{\log_2{n}-1}}{4}\right \rceil$.

We prove now that $G'$ and $H$ can be computed in $O(n)$ time. First, we compute the Schnyder realizer of $G$ in $O(n)$ time~\cite{Schnyder90}. Then, we compute the longest path
of each of the trees $T_0$, $T_1$, and $T_2$ and take the longest of
the three. If such a path $\pi$ has length at least $\left \lceil
\frac{\sqrt{\log_2{n}-1}}{4}\right \rceil$ we are in the first case
above. The computed path $\pi$ is a path $P_i(u)$ for some internal
vertex $u$ of $G$. $H$ is the subgraph induced by $P_i(u)$ and one
of the other two paths $P_{i-1}(u)$ and $P_{i+1}(u)$ (in this case
$G'$ coincides with $G$).  If $\pi$ is not long enough, we are in
the second case above. In this case we can find $G_\mu$ by computing
the $4$-block tree $T$, by performing a visit of $T$ in order to
find the node with the largest outdegree, and finally by computing
the pertinent graph of $\mu$. All this can be computed in
$O(n)$-time~\cite{jgaa/WangHe12}. Once we have $G_{\mu}$, we need to
compute a rectangular representation of it with the properties
described in the proof of Lemma~\ref{lemma:four_connected} and then
to compute the long stabber. A rectangular representation can be
computed in $O(n)$ time~\cite{DBLP:journals/talg/BuchsbaumGPV08} and
the stabber can also be computed in $O(n)$ time~\cite{Toth08}. In
this case $H$ is the subgraph induced by the path corresponding to
the stabber together with a vertex of the external boundary that is
adjacent to both endvertices of the path (in this case $G'$ is
$G_{\mu}$).
\end{proof}

%
%
%
%

\paragraph{Proof of Theorem~\ref{th:general}.}
Let $S$ be a set of distinct points in the plane such that $|S| \leq \left \lceil
\frac{\sqrt{\log_2 n}-1}{4} \right \rceil$ and let $G$ be an $n$-vertex plane graph. If $G$ is not maximal we can add dummy edges to it in order to make it maximal (these edges will be removed at end from the obtained drawing). Thus, without loss of generality, assume that $G$ is maximal. By Theorem~\ref{th:large-induced-outerpath}, in $G$ there exists a maximal plane subgraph $G'$ of $G$ that contains an external outerpath with a side path of size at least $\left \lceil \frac{\sqrt{\log_2 n}-1}{4} \right \rceil$. If $G'$ coincides with $G$, then, by Lemma~\ref{le:embedding.1.a}, $G$ admits a geometric point-subset embedding on $S$. If $G'$ does not coincide with $G$, then, by Lemma~\ref{le:embedding.1.a}, $G'$ admits a geometric point-subset embedding $\Gamma'$ on $S$. Since $G'$ is maximal plane, the external boundary of $G'$ is a triangle $t$ and it is drawn in the plane as a (geometric) triangle $\tau$. All the other vertices and edges of $G'$ are drawn inside $\tau$. If we remove from $G$ all the internal vertices of $G'$ we obtain a new maximal plane graph $G''$; the triangle $t$ is a separating triangle of $G$ and is an internal face of $G''$. Let $u$ be an internal vertex of $G'$ and let $\T$ be a Schnyder realizer of $G$. Since the three paths $P_0(u)$, $P_1(u)$, and $P_2(u)$ of $G$ are disjoint and $t$ is a separating triangle, each vertex of $t$ belong to exactly one of the three paths $P_0(u)$, $P_1(u)$, and $P_2(u)$, which are induced paths by Lemma~\ref{lemma:Schnyder_induced}. It follows that $G''$ satisfies the conditions of Lemma~\ref{le:embedding.1.b} and therefore it admits a geometric point-subset embedding $\Gamma''$ on $S$ that has $\tau$ has a subdrawing. The union of $\Gamma'$ and $\Gamma''$ is a geometric point-subset embedding of $G$ on $S$.

By Theorem~\ref{th:large-induced-outerpath}, Lemma~\ref{le:embedding.1.a}, and Lemma~\ref{le:embedding.1.b}, the time complexity of the drawing algorithm is $O(k \log k + n)$ where $k$ is the size of $S$. Since $k \leq \left \lceil \frac{\sqrt{\log_2 n}-1}{4} \right \rceil$, we have that $O(k \log k)=O(n)$ which concludes the proof of theorem~\ref{th:general}.

Let $F(n)$ be a function such that for every point set $S$ of size $F(n)$, any planar graph admits a geometric point-subset embedding on $S$. Angelini et al. proved that $4 \leq F(n) \leq 2 \lceil \frac{n}{3} \rceil+2$, and ask whether $\lim_{n \rightarrow + \infty} F(n) = + \infty$. An immediate consequence of Theorem~\ref{th:general} is the following.


\begin{corollary}\label{co:infty}
$\lim_{n \rightarrow + \infty} F(n) = + \infty$.
\end{corollary}

Furthermore, Theorem~\ref{th:large-induced-outerpath} implies the following corollary which is, for maximal planar graphs, an improvement over the result by Aroha and Valencia~\cite{ArohaV00} who showed that every $3$-connected planar graph with $n$ vertices and maximum vertex degree $d$ contains an induced path of length $\Omega(\sqrt{\log_3{d}})$.

\begin{corollary}\label{co:path}
Let $G$ be an $n$-vertex maximal plane graph. Then $G$ has an induced path with at least $\left \lceil \frac{\sqrt{\log_2 n}-1}{4} \right \rceil$ vertices.
\end{corollary}

\section{One-sided Convex Point Sets}\label{se:convex}

A \emph{one-sided convex point set $S$} is a set of points in the plane that are in convex position and can be rotated so that the leftmost point and the rightmost point are adjacent in the convex hull of $S$. This restriction on the point set allows us to increase the size of the set itself.
The structure of the proof of Theorem~\ref{th:convex} is similar to that of Theorem~\ref{th:general}. 

\subsection{Drawing tool}\label{ss:convex-embedding}

\begin{figure}
\centering
\includegraphics[scale=0.5]{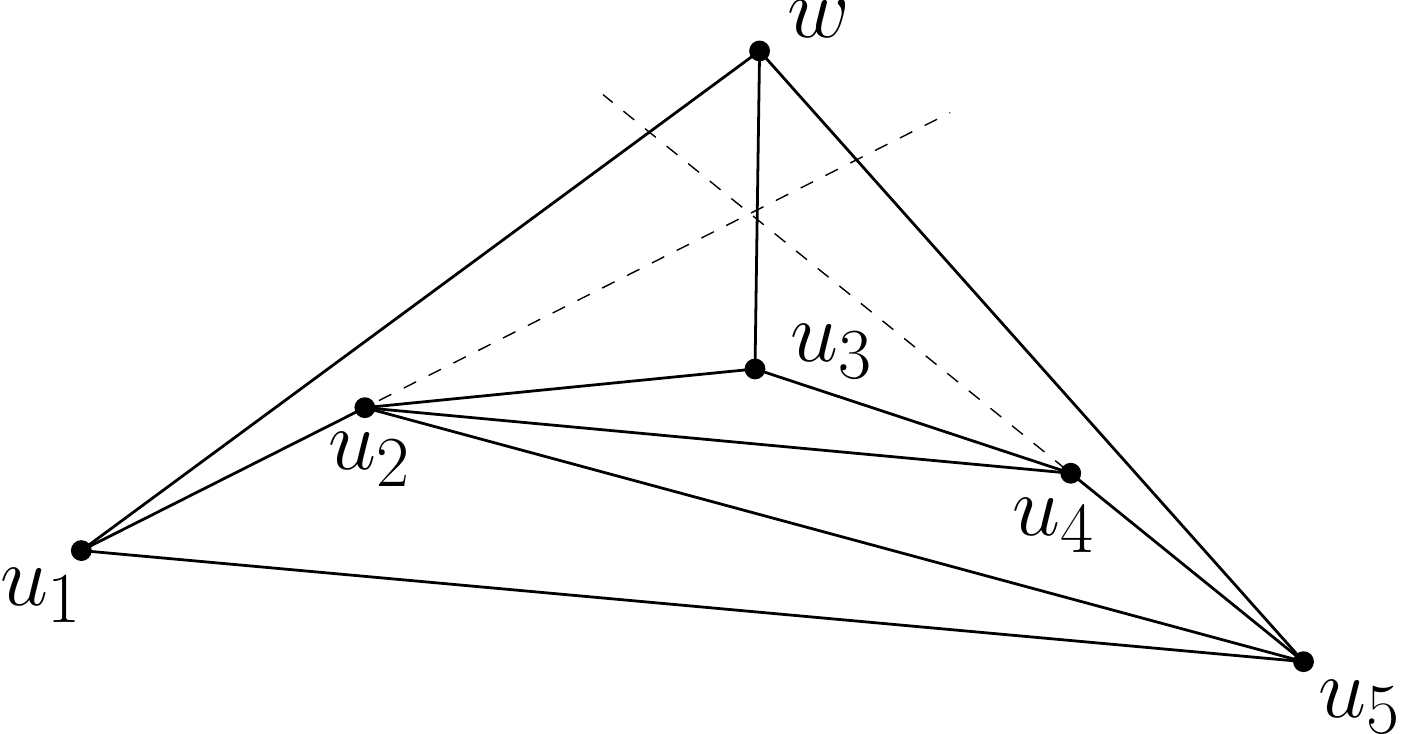}
\caption{\label{fi:drawing3}Illustration of the drawing technique of Lemma~\ref{le:embedding.2}.}
\end{figure}

\begin{lemma}\label{le:embedding.2}
Let $S$ be a one-sided convex point set with cardinality $k$. Let $G$ be an
$n$-vertex maximal plane graph that contains an induced biconnected outerplanar subgraph with at least $k$ vertices and one edge on the external boundary of $G$. Then $G$ admits a geometric point-subset embedding on $S$, which can be computed in $O(k \log k + n)$ time.
\end{lemma}
\begin{proof}
If necessary, rotate $S$ so that the leftmost and the rightmost point are adjacent in the convex hull of $S$. Let $p_1,p_2,\dots,p_k$ be the points of $S$ in the order they appear along the convex hull of $S$, with $p_1$ being the leftmost point and $p_k$ being the rightmost point. Without loss of generality assume that each point $p_i$ ($1 < i <k$) is above segment $\seg{p_1p_k}$. Let $H$ be the induced biconnected outerplanar subgraph of $G$ and let $C$ be external boundary of $H$. Without loss of generality assume that $H$ has exactly $k$ vertices (if not, we can augment $S$ with additional points so to maintain $S$ to be a one-sided convex position). So without loss of generality assume that $H$ has exactly $k$ vertices. Let $P$ be the path obtained from $C$ by removing the edge of $H$ that is on the external boundary of $G$. Let $u_1, u_2, \dots, u_k$ be the vertices of $P$ in the order they appear along $P$. Map $u_i$ to $p_i$ and draw each edge $(u_i,u_{i+1})$ as the segment $\seg{p_ip_{i+1}}$ (see Figure~\ref{fi:drawing3} for an illustration of the drawing technique). Also add every other edge connecting two of the vertices of $P$. Since $S$ is convex position, it is immediate to see that the resulting drawing has no crossings. Let $w$ be the outer vertex of $G$ that does not belong to $H$. Let $\rho_1$ be the straight line containing the segment $\seg{p_1p_2}$ and let $\rho_k$ be the straight line containing the segment $\seg{p_{k-1}p_{k}}$. Chose a point $p$ that is above $\rho_1$ and above $\rho_k$, place $w$ on $p$, and draw all edges connecting $w$ to the vertices of $H$. It is easy to see that the obtained drawing $\Gamma'$ has no crossing. Hence $\Gamma'$ is a planar geometric embedding on $S$ of the graph induced by the vertices in $H$ plus $w$. Moreover, all the faces of the drawing are star-shaped polygons. In the planar embedding of $G$ each vertex not yet drawn is inside one of the cycles represented by the faces of the computed drawing. Thus, the drawing can be completed by using Theorem~\ref{theorem:HongNagamochi} as described in the proof of Lemma~\ref{le:embedding.1.a}. Similar to Lemma~\ref{le:embedding.1.a}, the time spent to compute the drawing is $O(n)$, once the points are sorted. Thus, the time complexity is $O(k \log k +n)$.
\end{proof}

\subsection{Large Induced Biconnected Outerplanar Subgraph}\label{ss:convex-induced}

We now show that any $n$-vertex maximal plane graph contains an induced biconnected outerplanar subgraph with at least $ \lceil
\sqrt[3]{n} \rceil$ vertices and one edge on the external boundary of $G$. To this aim we will use the Schnyder realizer to define three different partial orders on the vertices of $G$.

Let $G$ be a maximal plane graph equipped with a Schnyder realizer $\T$. For each $i=1,2,3$, we define a partial order on $V(G)$ by constructing a graph $G_i$ by means of the following operations\footnote{Although defined in a different way, this graph is the same as the Frame graph defined in~\cite{BDHLMW-09}} (see Figure~\ref{fi:Gi} for an illustration).
\begin{itemize}
\item We first remove from $G$ all the internal edges with color different from $i$ and $i+1$;
\item We then reverse the orientation of the edges with color $i$;
\item We color the outer edge $(r_i,r_{i+1})$ with color $i$ and orient it from $r_i$ to $r_{i+1}$;
\item We color the outer edge $(r_i,r_{i-1})$ with color $i$ and orient it from $r_i$ to $r_{i-1}$;
\item We color the outer edge $(r_{i-1},r_{i+1})$ with color $i+1$ and orient it from $r_{i-1}$ to $r_{i+1}$.
\end{itemize}

The graph $G_i$ is a directed acyclic graph with one source $r_i$ and one sink $r_{i+1}$. Thus, it defines a partial order $\prec_{i}$ for the vertices of $G$; namely, $u \prec_{i} v$ if there is a directed path from $u$ to $v$ in $G_i$. Given a set $X$ with a partial order $\prec$ (also called a \textit{partial ordered set}), a subset $X'$ of $X$ such that every two elements in $X'$ are comparable with respect to $\prec$ is called a \textit{chain} of $X$. A subset $X''$ of $X$ such that no two elements in $X''$ are comparable with respect to $\prec$ is called an \textit{antichain} of $X$. According to the definition of $\prec_{i}$, a chain $\chi$ of $V(G)$ with respect to $\prec_{i}$ corresponds to a oriented path $\pi$ in $G_i$. We say that $\pi$ is the oriented path of $G_i$ \textit{associated with} $\chi$. We prove now some properties of the partial orders $\prec_i$.

\begin{figure}
\centering
\subfigure[$G_1$]{\includegraphics[width=0.32\columnwidth]{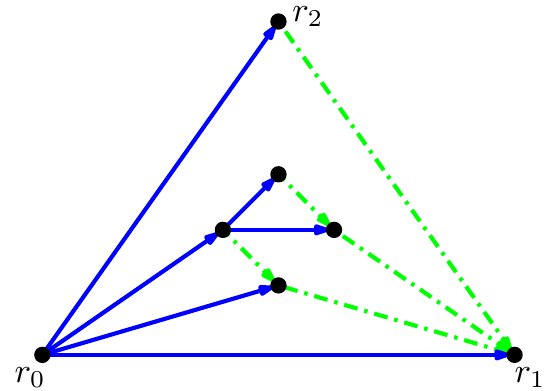}\label{fi:g1}}
\subfigure[$G_2$]{\includegraphics[width=0.32\columnwidth]{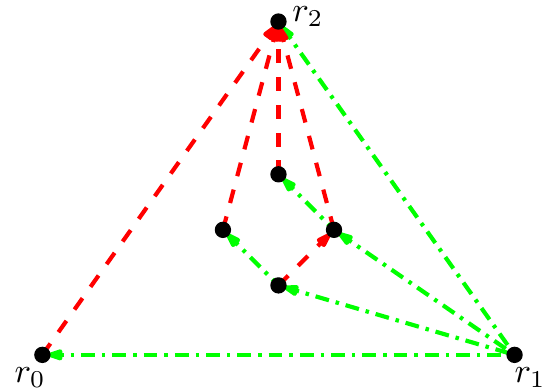}\label{fi:g2}}
\subfigure[$G_3$]{\includegraphics[width=0.32\columnwidth]{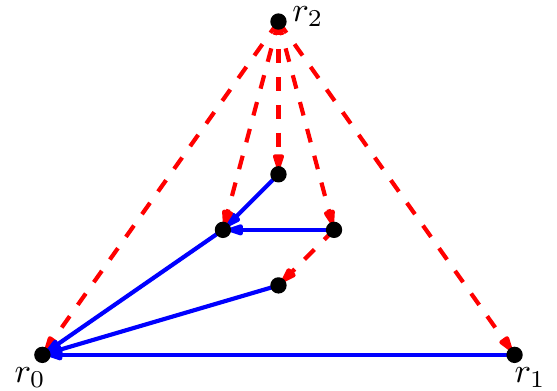}\label{fi:g3}}
\caption{The three directed graphs used to define the three partial orders $\prec_{1}$, $\prec_{2}$, and $\prec_{3}$.\label{fi:Gi}}
\end{figure}

\begin{lemma}\label{le:comparability}
Let $G$ be a maximal plane graph equipped with a Schnyder realizer $\T$. Let $u$ and $v$ be two internal vertices of $G$. Then $u$ and $v$ are comparable by at least one of the three partial orders $\prec_{1}$, $\prec_{2}$, and $\prec_{3}$.
\end{lemma}
\begin{proof}
If $u$ belongs to a $i$-path staring at $v$ for some $i$ ($1 \leq i \leq 3$), then $u$ and $v$ are comparable with respect to both $\prec_{i}$ and $\prec_{i-1}$ (indices modulo $3$). Namely, the concatenation of the reversed version of $P_i(v)$ with $P_{i+1}(v)$ is an oriented path of $G_i(v)$ containing both $u$ and $v$, and therefore $u$ and $v$ are comparable with respect to $\prec_i$. Analogously, the concatenation of the reversed version of $P_{i-1}(v)$ with $P_{i}(v)$ is an oriented path of $G_{i-1}(v)$ containing both $u$ and $v$, and therefore $u$ and $v$ are comparable with respect to $\prec_{i-1}$. Clearly, the same argument applies if $v$ belongs to a $i$-path staring at $u$ for some $i$ ($1 \leq i \leq 3$).

Consider the case when $u$ does not belong to any $i$-path staring at $v$ and viceversa. It follows that $u$ must be in one region $R_i(v)$ ($1 \leq i \leq 3$). Consider the $i$-path starting at $u$, $P_i(u)$. Since this path does not contain $v$, it must share a vertex $w$ with either $P_{i-1}(v)$ or $P_{i+1}(v)$ (see Figure~\ref{fi:comparability}). Suppose that $w$ belongs to $P_{i-1}(v)$, the case when it belongs to $P_{i+1}(v)$ is analogous. Consider the path $P$ consisting of the concatenation of: (i) the path $P_{i-1}(u)$; (ii) the subpath of $P_i(u)$ from $u$ to $w$; (iii) the subpath of $P_{i-1}(v)$ from $w$ to $v$; (iv) the path $P_i(v)$. All edges of $P$ colored $i-1$ are oriented towards $r_{i-1}$, while those colored $i$ are oriented towards $r_i$ (see Figure~\ref{fi:comparability2}). Thus reversing the edges of $P$ colored $i-1$, we obtain an oriented path $P'$ that is an oriented path of $G_{i-1}$ and that  contains both $u$ and $v$. Thus, $u$ and $v$ are comparable with respect to $\prec_{i-1}$.
\end{proof}

\begin{figure}
\centering
\subfigure[]{\includegraphics[width=0.32\columnwidth]{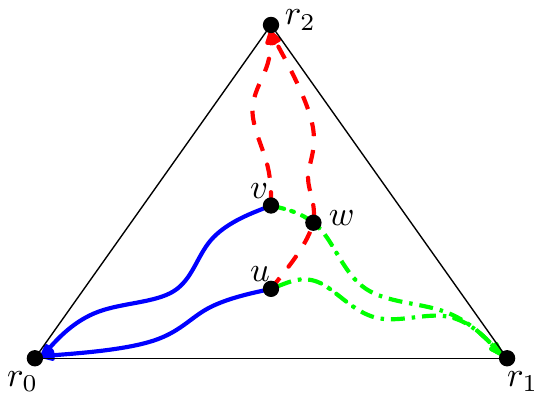}\label{fi:comparability}}
\hfil
\subfigure[]{\includegraphics[width=0.32\columnwidth]{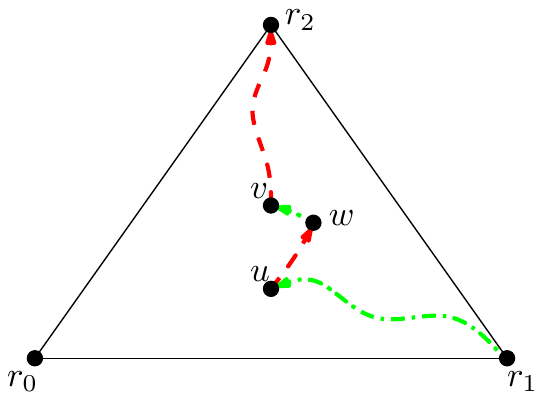}\label{fi:comparability2}}
\caption{An illustration for the proof of Lemma~\ref{le:comparability}. (a) The path $P_3(u)$ crosses the path $P_2(v)$. (b) The oriented path containing $u$ and $v$ in $G_2$. }
\end{figure}

\begin{lemma}\label{le:comparability2}
Let $G$ be a maximal plane graph equipped with a Schnyder realizer $\T$. Let $u$ be an internal vertex of $G$ and let $v$ be another vertex of $G$ such that $v \in R_{i-1}(u)$. Then $u$ and $v$ are not comparable with respect to $\prec_i$.
\end{lemma}
\begin{proof}
If $v \in R_{i-1}(u)$, the two paths $P_i(u)$ and $P_i(v)$ have a vertex in common $w_1$, and the paths $P_{i+1}(u)$ and $P_{i+1}(v)$ have a vertex in common $w_2$. Thus, in $G_i$, there are two oriented paths $\pi_1$ from $w_1$ to $w_2$: one, call it $\pi_1$, containing $u$ and the other one, call it $\pi_2$, containing $v$. We prove now that there is no oriented path from $v$ to $u$ in $G_i$ (the proof that there is no oriented path from $u$ to $v$ is analogous). Suppose that in $G_i$ there is an oriented path $\pi$ from $v$ to $u$. If $\pi$ does not share a vertex with $P_i(u) \cup P_{i+1}(u)$ nor with $P_i(v) \cup P_{i+1}(v)$, then it is completely contained in the (unoriented) cycle $C$ formed by $\pi_1$ and $\pi_2$. Let $e_i$ be the first edge of $P_i(u)$ (this edge is the unique outgoing edge of $u$ colored $i$ in $\T$) and let $e_{i+1}$ be the first edge of $P_{i+1}(u)$ (this edge is the unique outgoing edge of $u$ colored $i+1$ in $\T$). Since the edges of $\pi$ must be colored $i$ or $i+1$, there should be an edge colored $i$ or $i+1$ that follows $e_i$ and precedes $e_{i+1}$ in the counterclockwise order of the edges around $u$. But all edges that follows $e_i$ and precedes $e_{i+1}$ in the counterclockwise order around $u$ are colored $i-1$. Suppose now that $\pi$ shares a vertex $w$ with $P_i(u) \cup P_{i+1}(u)$ and/or $P_i(v) \cup P_{i+1}(v)$. Vertex $w$ cannot belong to path $P_i(u)$ because, in this case, $w$ would precede both $u$ and $v$ according to $\prec_i$ and therefore there would be a directed cycle in $G_i$. Analogously, $w$ cannot be in $P_{i+1}(v)$ because, in this case, $w$ would follow both $u$ and $v$ according to $\prec_i$ and therefore there would be a directed cycle in $G_i$. The only possibility is that $w$ belongs to either the subpath of $P_i(u)$ from $u$ to $w_1$, or to the subpath of $P_{i+1}(v)$ from $v$ to $w_2$ (endpoints excluded). Suppose that $w$ belongs to $P_i(u)$. Let $w'$ and $w''$ be the vertices that precede and follow $w$ along $P_i(u)$, respectively. In $\T$, edge $(w',w)$ is an incoming edge of $w$ colored $i$, while edge $(w,w'')$ is the unique outgoing edge of $w$ colored $i$. As a consequence, path $\pi$ cannot pass through $w$ ``exiting'' from cycle $C$. In order to exit from $C$, in $G_i$ there should be an incoming edge of $w$ that follows $(w',w)$ and precedes $(w,w'')$ in the counterclockwise order around $w$. In $\T$ this edge should be either a incoming edge of $w$ colored $i+1$ or an outgoing edge of $w$ colored $i$, but none of the two is possible according to the properties of the Schnyder realizer. With a symmetric argument one can show that also when $w$ belongs to $P_{i+1}(v)$, path $\pi$ cannot pass through $w$ ``exiting'' from cycle $C$. Thus, in both cases path $\pi$ can only enter cycle $C$. On the other hand, path $\pi$ cannot be completely outside cycle $C$ (in this case it would cross the common subpath of $P_i(u)$ and $P_i(v)$, or the common subpath of $P_{i+1}(u)$ and $P_{i+1}(v)$, but we have already seen that it is not possible). It follows that path $\pi$ enters $u$ from inside $C$, i.e., there is an edge colored $i$ or $i+1$ that follows $e_i$ and precedes $e_{i+1}$ in the counterclockwise order of the edges around $u$, but we have already seen that this is not possible.
\end{proof}

\begin{lemma}\label{le:comparability3}
Let $G$ be a maximal plane graph equipped with a Schnyder realizer $\T$. Let $u$ be an internal vertex of $G$ and let $v$ be another vertex of $G$ that is comparable with $u$ by $\prec_i$. If $v \prec_i u$ then $v \in R_{i+1}(u) \cup P_i(u)$. If $u \prec_i v$ then $v \in R_{i}(u) \cup P_{i+1}(u)$.
\end{lemma}
\begin{proof}
We consider the case when $v \prec_i u$ (the case when $u \prec_i v$ can be proven symmetrically). We show that $v$ must be in $R_{i+1}(u) \cup P_i(u)$. By Lemma~\ref{le:comparability2}, $v$ cannot be in $R_{i-1}(u)$ nor in $P_{i-1}(u)$ (in this case $u$ would be in $R_{i-1}(v)$). If $v$ is in $P_{i+1}(u)$ then $u$ ad $v$ are comparable with respect to $\prec_i$ but $u \prec_i v$, so $v$ cannot be in $P_{i+1}(u)$. Suppose that $v$ is in $R_{i}(u)$. If $P_{i}(v)$ has a vertex in common with $P_{i-1}(u)$, then we are again in the case when $u$ is in $R_{i-1}(v)$ and thus $u$ and $v$ are not comparable by Lemma~\ref{le:comparability2}. It follows that $P_{i}(v)$ must have a vertex $w$ in common with $P_{i}(u)$. Consider the path $P$ consisting of the concatenation of: (i) the path $P_{i}(u)$; (ii) the subpath of $P_{i+1}(u)$ from $u$ to $w$; (iii) the subpath of $P_{i}(v)$ from $w$ to $v$; (iv) the path $P_{i+1}(v)$. All edges of $P$ colored $i$ are oriented towards $r_i$, while those colored $i+1$ are oriented towards $r_{i+1}$. Thus reversing the edges of $P$ colored $i$, we obtain an oriented path $P'$ that is an oriented path of $G_{i}$ and that contains both $u$ and $v$. Thus, $u$ and $v$ are comparable with respect to $\prec_{i}$. Note however that in this case $u \prec_i v$. Thus $v$ cannot be in $R_{i}(u)$. The only possibility remaining is that $v$ is in $R_{i+1}(u) \cup P_i(u)$.
\end{proof}

In the next lemma we show that any maximal chain of $V(G)$ with respect to any partial order $\prec_{i}$ defines an induced  biconnected outerplanar subgraph with one edge on the external boundary of $G$.

\begin{lemma}\label{le:chain-frontier}
Let $G$ be a maximal plane graph equipped with a Schnyder realizer $\T$. Let $\chi$ be a maximal chain of $V(G)$ with respect to the partial order $\prec_{i}$, $1 \leq i \leq 3$ and let $\pi$ be the oriented path associated with $\chi$. The graph induced by the vertices of $\pi$ is an induced biconnected outerplanar subgraph with one edge on the external boundary of $G$.
\end{lemma}
\begin{proof}
First of all, notice that since $G_i$ has a single source $r_i$ and single sink $r_{i+1}$, then every maximal chain of $V(G)$ with respect to $\prec_{i}$ has $r_i$ as the first element and $r_{i+1}$ as the last element. As a consequence $\pi$ starts at $r_i$ and ends at $r_{i+1}$. Also, since $\chi$ contains at least one internal vertex, $\pi$ consists of all internal edges. This implies that it does not coincide with the single edge $(r_i,r_{i+1})$. Thus, the concatenation of (the underlying undirected path of) $\pi$ and of the edge $(r_i,r_{i+1})$ is a simple cycle $C$ of $G$. We prove that the subgraph $H$ induced by the vertices in $C$ is a biconnected outerplanar graph. If $C$ is the boundary of the external face, then $H$ is a $3$-cycle and therefore it is a biconnected outerplanar graph. If $C$ is not the boundary of the external face of $G$, in order to prove that is a biconnected outerplanar graph, it is sufficient to show that all edges of $H$ that are not in $C$ are inside $C$ in the planar embedding of $G$. Consider an edge $(u,v)$ of $H$ not in $C$. Vertices $u$ and $v$ must be two non-consecutive vertices of $\pi$. Without loss of generality assume $u \prec_i v$, i.e., $u$ is encountered first when walking along $\pi$ from $r_i$ to $r_{i+1}$. Since we are considering an edge not in $C$, it must be $u \neq r_i$ or $v \neq r_{i+1}$ (or both). Since $u$ and $v$ are non-consecutive, there must be a vertex $w$ such that $u \prec_i w \prec_i v$. By Lemma~\ref{le:comparability3}, $u$ is in $R_{i+1}(w) \cup P_i(w)$ and $v$ is in $R_{i}(w) \cup P_{i+1}(w)$. As a consequence, if edge $(u,v)$ is not inside $C$, then it crosses path $P_{i-1}(w)$, which is not possible.
\end{proof}

The next lemma shows that the subgraph of Lemma~\ref{le:chain-frontier} has at least $\lceil \sqrt[3]{n} \rceil$ vertices.

\begin{lemma}\label{le:chain}
Let $G$ be a maximal plane graph equipped with a Schnyder realizer $\T$. $V(G)$ has a chain of size at least $\lceil \sqrt[3]{n} \rceil$ with respect to one of the three partial orders $\prec_{1}$, $\prec_{2}$, and $\prec_{3}$.
\end{lemma}
\begin{proof}
Notice that, in order to prove that $V(G)$ has a chain of size at least $\lceil \sqrt[3]{n} \rceil$, it is sufficient to show that $V(G)$ has a chain of size at least $\sqrt[3]{n}$ because the number of vertices in a chain is necessarily an integer number. Consider first the partial order $\prec_{1}$. If $V(G)$ has a chain of size at least $\sqrt[3]{n}$ with respect to $\prec_{1}$, then the statement is true. If not, by Mirsky's theorem~\cite{M-71} (the dual of Dilworth's theorem), $V(G)$ has a partition into at most $\sqrt[3]{n}$ antichains with respect to $\prec_{1}$. Hence, one of this antichains must have at least $\frac{n}{\sqrt[3]{n}}=n^{2/3}$ vertices. Let $U$ be such a set of vertices. Since the vertices in $U$ form an antichain of $V(G)$ with respect to $\prec_{1}$, then, by Lemma~\ref{le:comparability}, any two vertices in $U$ must be comparable by $\prec_{2}$ or by $\prec_{3}$. Consider $\prec_{2}$ restricted to the set $U$; if $U$ has a chain of size at least $\sqrt[3]{n}$ with respect to $\prec_{2}$, then the statement is true. Otherwise, applying again Mirsky's theorem, $U$ can be partitioned into at most $\sqrt[3]{n}$ antichains. One of this antichain must have at least $\frac{n^{2/3}}{n^{1/3}}=\sqrt[3]{n}$ vertices. Let $U'$ be such a set of vertices. The vertices in $U'$ are not comparable with respect to $\prec_{1}$, neither with respect to $\prec_{2}$. Thus by Lemma~\ref{le:comparability} they must be comparable by $\prec_{3}$, i.e., they form a chain of $U'$ with respect to $\prec_{3}$.
\end{proof}

\begin{theorem}\label{th:large-induced-outerplanar}
Every $n$-vertex maximal plane graph $G$ contains an induced biconnected outerplanar subgraph with at least $\lceil \sqrt[3]{n} \rceil$ vertices and one edge on the external boundary of $G$. Such a subgraph can be computed in $O(n)$-time.
\end{theorem}
\begin{proof}
Consider a Schnyder realizer $\T$ of $G$ and define the three partial orders $\prec_{1}$, $\prec_{2}$, and $\prec_{3}$. By Lemma~\ref{le:chain}, $V(G)$ has a chain of size at least $\lceil \sqrt[3]{n} \rceil$ with respect to one of these partial orders. From Lemma~\ref{le:chain-frontier} it follows that $G$ contains an induced biconnected outerplanar subgraph $H$ with at least $\lceil \sqrt[3]{n} \rceil$ vertices and one edge on the external boundary of $G$.

The Schnyder realizer $\T$ can be computed in $O(n)$ time~\cite{Schnyder90}. Using $\T$, the three graphs $G_1$, $G_2$, and $G_3$ can be also computed in $O(n)$ time. Computing the longest path in each of these graphs (which can be done in $O(n)$ time) we find the longest chain of $V(G)$ with respect to one of the three partial orders $\prec_{1}$, $\prec_{2}$, and $\prec_{3}$. The subgraph induced by the vertices in this path is $H$.
\end{proof}

\paragraph{Proof of Theorem~\ref{th:convex}} Let $S$ be a one-sided convex point set of cardinality at most $\lceil \sqrt[3]{n} \rceil$ and let $G$ be an $n$-vertex plane graph. If $G$ is not maximal we can add dummy edges to it in order to make it maximal (these edges will be removed at end from the obtained drawing). Thus, without loss of generality, assume that $G$ is maximal. By Theorem~\ref{th:large-induced-outerplanar} $G$ contains an induced biconnected outerplanar subgraph with at least $\lceil \sqrt[3]{n} \rceil$ vertices and one edge on the external boundary of $G$. By Lemma~\ref{le:embedding.2}, $G$ admits a geometric point-subset embedding on $S$.

By Theorem~\ref{th:large-induced-outerplanar} and Lemma~\ref{le:embedding.2} the time complexity of the drawing algorithm is $O(k \log k + n)$ where $k$ is the size of $S$. Since $k \leq \lceil \sqrt[3]{n} \rceil$, we have that $O(k \log k)=O(n)$ which concludes the proof of Theorem~\ref{th:convex}.

\medskip


Theorem~\ref{th:convex} can be used to solve a different problem,
called the \emph{allocation problem}. In the allocation
problem~\cite{rv-cssld-11} the input is an $n$ vertex planar graph
$G$ and a point set $X$ of cardinality $n$, and the goal is to
compute a planar straight-line drawing of $G$ such that as many
vertices as possible are represented by points of $X$. The number of points fitting this purpose is denoted as $fit^X(G)$ and $fit(G)= \min_X\{fit^{X}(G)\}$, where the minimum is taken over all $n$-point sets $X$. Upper and lower bounds on
$fit(G)$ are known for various classes of
graphs~\cite{BDHLMW-09,DBLP:journals/dcg/GoaocKOSSW09,rv-cssld-11}. We define $fit^C(G)= \min_X\{fit^{X}(G)\}$ where the minimum is taken over all $X$ in one-sided convex position. Theorem~\ref{th:convex} implies that $fit^C(G) \geq \lceil \sqrt[3]{n}
\rceil$ for every planar graph $G$. This improves the
best known lower bound on $fit(G)$ for general planar graphs which is
$\sqrt[4]{\frac{n}{3}}$~\cite{BDHLMW-09}.

\begin{corollary}
For every $n$-vertex planar graph, $fit^{C}(G) \geq  \lceil \sqrt[3]{n} \rceil$.
\end{corollary}
\section{Conclusions and Open Problems}\label{se:conclusions}

In this paper we proved that a planar graph with $n$ vertices admits
a geometric point-subset embedding on every set of $\left \lceil \frac{\sqrt{\log_2 n}-1}{4}
\right \rceil$ points and on every one-sided convex point set of $\left \lceil \sqrt[3]{n} \right \rceil$  points on the plane.
These are the first non-constant lower bounds for the size of arbitrary (one-sided convex) point sets on which planar graphs have geometric point-subsets embedding.
In order to achieve this result we proved that every maximal planar graph contains as a subgraph an induced outerpath of size at least $\left \lceil \frac{\sqrt{\log_2 n}-1}{4} \right \rceil$ and
an induced biconnected outerplanar graph of size $\left \lceil \sqrt[3]{n} \right \rceil$. To the best of our knowledge, these are the first lower bounds for the size of a largest induced outerpath and biconnected outerplanar subgraph in
planar graphs.
It is known that not all planar graphs with $n$ vertices admit a geometric point-subset embedding on a one-sided convex set containing more
than $2 \left \lceil \frac{n}{3} \right \rceil + 2$ points~\cite{DBLP:conf/isaac/Angelini12}. Thus closing the gap between lower bounds $\left \lceil \frac{\sqrt{\log_2 n}-1}{4}
\right \rceil$ and $\left \lceil \sqrt[3]{n} \right \rceil$,  and  upper bound $2 \left \lceil \frac{n}{3} \right \rceil + 2$ is an interesting open problem.
From the graph-theoretical point of view it would be also interesting to study the upper bounds for the size of maximum induced biconnected outerplanar subgraph and induced outerpath in biconnected planar graphs. However, it is worth mentioning that an upper bound for the size of induced outerplanar subgraph does not provide us an upper bound for the size of convex point set that represent universal subset for planar graphs.

\end{document}